\documentclass[11pt]{amsart}
\usepackage[utf8]{inputenc}
\usepackage[T1]{fontenc}
\usepackage{amsfonts,amsmath,amssymb,amsthm,bm}
\usepackage{enumerate}
\usepackage{caption}
\usepackage{hyperref}
\usepackage{multirow}
\usepackage{manfnt}
\usepackage{color}
\usepackage{graphics}
\usepackage{tikz}
\usetikzlibrary{arrows.meta}
\usetikzlibrary{decorations.pathreplacing}

\newtheorem{Theorem}{Theorem}
\newtheorem{Definition}[Theorem]{Definition}

\newtheorem{Lemma}[Theorem]{Lemma}

\newtheorem{Proposition}[Theorem]{Proposition}
\newtheorem{Corollary}[Theorem]{Corollary}
\newtheorem{Remark}[Theorem]{Remark}
\newtheorem*{Example}{Example}

\newcommand{\N}{\mathbb{N}}

\newcommand{\LL}{\mathsf{LL}}
\DeclareMathOperator{\sing}{\mathsf{Sing}}

\def\alphabet{\Sigma}

\begin{document}

\title{On the binomial equivalence classes of finite words}
\author{Marie Lejeune}
\thanks{The first author is supported by a FNRS fellowship}
\author{Michel Rigo}
\author{Matthieu Rosenfeld}
\email{\{M.Lejeune,M.Rigo\}@uliege.be; matthieu.rosenfeld@gmail.com}
\address{University of Li\`ege, 
Dept. of Mathematics, 
All\'ee de la d\'ecouverte 12 (B37), 
B-4000 Li\`ege, 
Belgium}

\begin{abstract}
Two finite words $u$ and $v$ are $k$-binomially equivalent if, for each word $x$ of length at most $k$, $x$ appears the same number of times as a subsequence (i.e., as a scattered subword) of both $u$ and $v$. This notion generalizes abelian equivalence. In this paper, we study the equivalence classes induced by the $k$-binomial equivalence with a special focus on the cardinalities of the classes. We provide an algorithm
generating the $2$-binomial equivalence class of a word. For $k\ge 2$ and alphabet of $3$ or more symbols, the language made  of lexicographically least elements  of every $k$-binomial equivalence class and the language of singletons, i.e., the words whose $k$-binomial equivalence class is restricted to a single element, are shown to be non context-free. As a consequence of our discussions,  we also prove that the submonoid generated by the generators of the free nil-$2$ group on $m$ generators is isomorphic to the quotient of the free monoid $\{1,\ldots,m\}^*$ by the $2$-binomial equivalence.
\end{abstract}

\maketitle


\smallskip
\noindent \textbf{68R15, 68Q45, 05A05, 20F18}

\section{Introduction}

Let $\alphabet$ be a totally ordered alphabet of the form $\{1<\cdots<m\}$. We make use of the same notation $<$ for the induced lexicographic order on $\alphabet^*$. 

Let $\sim$ be an equivalence relation on $\alphabet^*$. The equivalence class of the word $w$ is denoted by $[w]_\sim$. We will be particularly interested in two types of subsets of $\alphabet^*$ with respect to $\sim$. We let 
$$\LL(\sim,\alphabet)=\left\{w\in\alphabet^*\mid \forall u\in [w]_\sim : w\leq u \right\}$$
denote the language of lexicographically least elements of every equivalence class for $\sim$. So there is a one-to-one correspondence between $\LL(\sim,\alphabet)$ and $\alphabet^*\!/\!\sim$. 
We let
$$\sing(\sim,\alphabet)=\left\{w\in\alphabet^*\mid \# [w]_\sim=1\right\}$$
denote the language made of the so-called {\em $\sim$-singletons}, i.e., the elements whose equivalence class is restricted to a single element. Clearly,  we have $\sing(\sim,\alphabet)\subseteq\LL(\sim,\alphabet)$. In the extensively studied context of Parikh matrices (see Section~\ref{sec:22}), two words are {\em $M$-equivalent} if they have the same Parikh matrix. In that setting, singletons are usually called {\em $M$-unambiguous words} and have attracted the attention of researchers, see, for instance, \cite{Salomaa} and the references therein. 
\medskip

Let $k\ge 1$ be an integer. Let $\sim_{k,ab}$ be the $k$-abelian equivalence relation introduced by Karhum\"aki \cite{Kar}. Two words are {\em $k$-abelian equivalent} if they have the same number of factors of length at most $k$. 
If $k=1$, the words are {\em abelian equivalent}. We denote by $\Psi(u)$ the {\em Parikh vector} of the finite word~$u$, defined as
$$
\Psi(u) = \left( |u|_1, \ldots, |u|_m \right),
$$
where $|u|_a$ is the number of occurrences of the letter $a$ in $u$.
Two words $u$ and $v$ are abelian equivalent if and only if $\Psi(u) = \Psi(v)$. 

The $k$-abelian equivalence relation has recently received a lot of attention, see, for instance, \cite{KS1,KS2}. In particular, the number of $k$-abelian singletons of length $n$ is studied in \cite{singleton}. Based on an operation of $k$-switching, the following result is given in \cite{CKPW}.
\begin{Theorem}\label{thm:abel}
Let $k\ge 1$. Let $\alphabet$ be a $m$-letter alphabet. 
    For the $k$-abelian equivalence, the two languages $\LL({\sim_{k,ab}},\alphabet)$ and $\sing({\sim_{k,ab}},\alphabet)$ are regular. 
\end{Theorem}

As discussed in Section~\ref{sec:22}, the set of $M$-unambiguous words over a $2$-letter alphabet is also known to be regular. Motivated by this type of results, we will consider another equivalence relation, namely the $k$-binomial equivalence introduced in \cite{RigoSalimov:2015}, and study the corresponding sets $\LL$ and $\sing$. 

\begin{Definition}
    We let the {\em binomial coefficient} $\binom{u}{v}$ denote the number of times $v$ appears as a (not necessarily contiguous) subsequence of~$u$. Let $k\ge 1$ be an integer. Two words $u$ and $v$ are {\em $k$-binomially equivalent}, denoted $u\sim_k v$, if $\binom{u}{x}=\binom{v}{x}$ for all words $x$ of length at most $k$.
\end{Definition}

We will show that $k$-abelian and $k$-binomial equivalences have incomparable properties for the corresponding languages $\LL$ and $\sing$. These two equivalence are both a refinement of the classical abelian equivalence and it is interesting to see how they differ. As mentioned by Whiteland in his Ph.D. thesis: ``part of the challenges in this case follow from the property that a modification in just one position of a word can have global effects of the distribution of subwords, and thus the structure of the equivalence classes.'' \cite{Whit}.

This paper is organized as follows. The special case of $2$-binomial equivalence over a $2$-letter alphabet is presented in Section~\ref{sec:22}: the corresponding languages are known to be regular. In Section~\ref{sec:3}, we discuss an algorithm generating the $2$-binomial equivalence class of any word over an arbitrary alphabet. Then we prove that the submonoid generated by the generators of the free nil-$2$ group on $m$ generators is isomorphic to $\{1,\ldots,m\}^*\!/\!\sim_2$. Section~\ref{sec:4} is about the growth rate of $\#(\alphabet^n\!/\!\sim_k)$.  As a consequence of Sections~\ref{sec:3} and \ref{sec:4},  
the growth function for the submonoid generated by the generators of the free nil-$2$ group on $m\ge 3$ generators is in $\Theta\left(n^{m^2-1}\right)$. In the last section, contrasting with Theorem~\ref{thm:abel}, we show that $\LL(\sim_k,\alphabet)$ and $\sing(\sim_k,\alphabet)$ are rather complicated languages when $k\ge 2$ and $\#\alphabet\ge 3$: they are not context-free.

\section{$2$-binomial equivalence over a $2$-letter alphabet}\label{sec:22}

Let $\alphabet=\{1,2\}$ be a $2$-letter alphabet. Recall that the {\em Parikh matrix} associated with a word $w\in\{1,2\}^*$ is the $3\times 3$ matrix given by
$$
P(w)=\begin{pmatrix}
    1 & |w|_{1} & \binom{w}{12} \\
    0 & 1 &  |w|_{2} \\
    0 & 0 & 1 \\
\end{pmatrix}.$$
For $a,b\in\{1,2\}$, $\binom{w}{ab}$ can be deduced from $P(w)$. Indeed, we have
$\binom{w}{aa}=\binom{|w|_a}{2}$ and if $a\neq b$, 
\begin{equation}
    \label{eq:remobv}
\binom{w}{aa}+\binom{w}{ab}+\binom{w}{ba}+\binom{w}{bb}=\binom{|w|_a+|w|_b}{2}.    
\end{equation}
It is thus clear that $w\sim_2 x$ if and only if $P(w)=P(x)$. We can therefore make use of the following theorem of Foss\'e and Richomme \cite{FR}. 
If two words $u$ and $v$ over an arbitrary alphabet $\alphabet$ can be factorized as $u=xabybaz$ and $v=xbayabz$ with $a,b \in \alphabet$, we write $u\equiv_2 v$. The reflexive and transitive closure of this relation is denoted by $\equiv_2^*$. 

\begin{Theorem}
\label{thm:fosse}
 Let $u,v$ be two words over $\{1,2\}$. The following assertions are equivalent:
    \begin{itemize}
      \item the words $u$ and $v$ have the same Parikh matrix;
      \item the words $u$ and $v$ are $2$-binomially equivalent;
      \item $u\equiv_2^* v$.
    \end{itemize}
\end{Theorem}
Consequently, the language $\sing({\sim_{2}},\{1,2\})$ avoiding two separate occurrences of $12$ and $21$ (or, $21$ and $12$) is regular. A regular expression for this language is given by
$$1^*2^*+2^*1^*+1^*21^*+2^*12^*+1^*212^*+2^*121^*.$$
A NFA accepting $\sing({\sim_{2}},\{1,2\})$ was given in \cite{Salomaa}.

\begin{Remark}
From \cite{RigoSalimov:2015}, we know that $$\# \LL({\sim_{2}},\{1,2\})= \# \left(\{1,2\}^n/\!\sim_2 \right)=\frac{n^3+5n+6}{6}.$$
Note that this is exactly the sequence {\tt A000125} of {\em cake numbers}, i.e., the  maximal number of pieces resulting from $n$ planar cuts through a cube.
\end{Remark}

\begin{Proposition}
    The language $\LL({\sim_{2}},\{1,2\})$ is regular.
\end{Proposition}

\begin{proof}
    As a consequence of Theorem~\ref{thm:fosse}, if a word $u$ belongs to $\LL({\sim_{2}},\{1,2\})$, it cannot be of the form $x21y12z$ because otherwise, the word $x12y21z$ belongs to the same class and is lexicographically less. 
	Consequently,
	$$\LL({\sim_{2}},\{1,2\}) \subseteq \{1,2\}^*\setminus \{1,2\}^* 21 \{1,2\}^* 12 \{1,2\}^*.$$
The reader can check that the language in the r.h.s. has exactly $(n^3+5n+6)/6$ words of length $n$. We conclude with the previous remark that the two languages are thus equal. 
\end{proof}

\section{$2$-binomial equivalence over a $m$-letter alphabet}\label{sec:3}

Theorem~\ref{thm:fosse} does not hold for ternary or larger alphabets. Indeed, the two words $1223312$ and $2311223$ are $2$-binomially equivalent but both words belong to $\sing(\equiv_2, \{1,2,3\})$ which means that $1223312\not\equiv_2^* 2311223$. It is therefore meaningful to study $\sim_2$ over larger alphabets and to describe the $2$-binomial equivalence classes.
\medskip

The first few terms of $\left(\# \left(\{1,2,3\}^n/\!\sim_2\right)\right)_{n\ge 0}$ are given by $$1, 3, 9, 27, 78, 216, 568, 1410,\ldots.$$ This sequence also appears in the Sloane's encyclopedia as entry {\tt A140348} which is the growth function for the submonoid generated by the generators of the free nil-$2$ group on three generators. In this section, we make explicit the connection between these two notions (see Theorem~\ref{thm:isom}).

Recall that the {\em commutator} of two elements $x,y$ belonging to a multiplicative group $(G,\cdot)$ is $[x,y]=x^{-1}y^{-1}xy$. Hence, the following relations hold
$$xy=yx[x,y]\quad \forall x,y\in G.$$
A {\em nil-$2$} group is a group $G$ for which the commutators belong to the center $Z(G)$, i.e., 
\begin{align}
\label{equ:comm}
[x,y]z=z[x,y]\quad \forall x,y,z\in G.
\end{align}
Let $\alphabet=\{1,\ldots,m\}$. The free nil-$2$ group on $m$ generators has thus a presentation
$$N_2(\alphabet)= \left\langle \alphabet \, \mid\,   [x,y]z=z[x,y]\ (x,y,z\in \alphabet)\right\rangle.$$
As an example, making use of these relations, let us show that two elements of the free group on $\{1,2,3\}$ are equivalent in $N_2(\{1,2,3\})$:
$$12321=(12[2,1])[1,2]321=21[1,2]321=213(21[1,2])=21312.$$

Let $\alphabet^{-1}$ be the alphabet $\{1^{-1}, \ldots, m^{-1} \}$ of the inverse letters, that we suppose disjoint from $\alphabet$. By abuse of notation, for all $x\in\alphabet$, $(x^{-1})^{-1}$ is the letter $x$. 
Since $N_2(\alphabet)$ is the quotient of the free monoid $\left( \alphabet \cup \alphabet^{-1} \right)^*$ under the congruence relations generated by $xx^{-1} = \varepsilon$ and \eqref{equ:comm}, we will consider the natural projection denoted by
$$\pi : \left(\alphabet \cup \alphabet^{-1} \right)^* \rightarrow N_2(\alphabet).
$$ 

In Section~\ref{ss31}, we provide an algorithmic description of any $2$-binomial class. We make use of this description in Section~\ref{ss32} to show that the monoid $\alphabet^*/\!\sim_2$ is isomorphic to the submonoid, generated by $\alphabet$, of the nil-$2$ group $N_2(\alphabet)$.

\subsection{A nice tree}\label{ss31}

%For the sake of notation, assume that we have a $3$-letter alphabet. 
Let $w$ be a word over $\alphabet$ and  $\ell$ the lexicographically least element in its abelian equivalence class, i.e., $$\ell=1^{|w|_1}2^{|w|_2}\cdots m^{|w|_m}.$$ Consider the following algorithm that, given $\ell$, produces the word $w$  only by exchanging adjacent symbols. We let $\wedge(u,v)$ denote the longest common prefix of two finite words $u$ and $v$. We define a sequence of words $\ell_i$ starting with $\ell_0=\ell$.

\medskip

$i\leftarrow 0$, $\ell_0\leftarrow\ell$

{\bf while} $\ell_i\neq w$ 

\quad $p\leftarrow \wedge(\ell_i,w)$ thus $\ell_i=pcx$ and $w=pdy$ with $c,d \in \alphabet$ and $c<d$

\quad consider the leftmost $d$ occurring in $cx$, i.e, 

\quad $cx=cudv$ and $u$ only contains letters less than $d$. 

\quad {\bf for} $j=1$ to $|u|$

\quad\quad $\ell_{i+j}\leftarrow p c u_1\cdots u_{|u|-j} d u_{|u|-j+1} \cdots u_{|u|} v$

\quad $\ell_{i+|u|+1}\leftarrow p d c u v$

\quad $i \leftarrow i + |u| + 1$
\medskip

\begin{Remark}
It can easily be shown that, at the beginning of each iteration of the while loop, the word $cx$ is the least lexicographic word of its abelian class. It follows that $c < d$ and $u$ only contains letters less than $d$.
\end{Remark}

In the for loop, the two letters $u_{|u|-j+1}$ and $d$ are exchanged. Observe that $u_{|u|-j+1}<d$. After the for loop, the two letters $c$ and $d$ are exchanged (and again $c<d$). We record the sequence and number of exchanges of the form $cd\mapsto dc$ for all $c,d\in\alphabet$ with $c<d$ that are performed when executing the algorithm.

Using~\eqref{eq:remobv}, the next lemma is obvious. 
\begin{Lemma}
    Two abelian equivalent words $u,v$ are $2$-binomially equivalent if and only if $\binom{u}{ab}=\binom{v}{ab}$ for all $a,b\in\alphabet$ with $a<b$. Let $\ell$ be a word in $1^* 2^*\cdots m^*$. In the set of tuples of size $m(m-1)/2$
$$\left\{
\left(\binom{w}{12},\ldots,\binom{w}{1m},\binom{w}{23},\ldots, \binom{w}{2m},\ldots , \binom{w}{(m-1)m}\right) \mid w\in [\ell]_{\sim_1} \right\},$$
the greatest element, for the lexicographic ordering, is achieved for $w=\ell$. 

% The word $\ell=1^{n_1}2^{n_2}\cdots m^{n_m}$ has the greatest tuple of size $m(m-1)/2$
% $$\left(\binom{\ell}{12},\ldots,\binom{\ell}{1m},\binom{\ell}{23},\ldots, \binom{\ell}{2m},\ldots , \binom{\ell}{(m-1)m}\right)$$
% for the lexicographic ordering in its abelian equivalence class. It is the only word whose tuple is given by 
% $$(n_1\cdot n_2,\ldots,n_1\cdot n_m,n_2\cdot n_3,\ldots, n_2 \cdot n_m,\ldots, n_{m-1}\cdot n_m).$$
\end{Lemma}

We consider the $m(m-1)/2$ coefficients $\binom{u}{ab}$ with $a<b$. Note that, in the algorithm, if $\ell_{j+1}$ is obtained from $\ell_j$ by an exchange of the form $ab\mapsto ba$, all these coefficients remain unchanged except for
\begin{align}\label{eq:moinsUn}
\binom{\ell_{j+1}}{ab}=\binom{\ell_{j}}{ab}-1.
\end{align}

\begin{Corollary}
    When applying the algorithm producing the word $w$ from the word $\ell=1^{|w|_1}2^{|w|_2}\cdots m^{|w|_m}$, the total number of exchanges $ab\mapsto ba$, with $a < b$, is given by
$$\binom{\ell}{ab}-
\binom{w}{ab} = \binom{w}{ba}.$$
\end{Corollary}

Consequently two words are $2$-binomially equivalent if and only if they are abelian equivalent and the total number of exchanges of each type $ab \mapsto ba$, $a<b$, when applying the algorithm to these two words, is the same. An equivalence class $[w]_{\sim_2}$ is thus completely determined by a word $\ell=1^{n_1}2^{n_2}\cdots m^{n_m}$ and the numbers of different exchanges. We obtain an algorithm generating all words of $[w]_{\sim_2}$.

\begin{Definition}
Let us build a (directed) tree whose vertices are words, the root is $\ell_0 = 1^{|w|_1}2^{|w|_2}\cdots m^{|w|_m}$. There exists an edge from $v$ to $v'$ if and only if $v'$ is obtained by an exchange of the type $ab \mapsto ba$, $a<b$, from $v$. The edge is labeled with the applied exchange.
\end{Definition}
 To generate the $\sim_2$-equivalence class of $w$, it suffices to take all the nodes that are at level $\sum_{1\leq a < b \leq m} \binom{w}{ba}$ such that the path from $\ell_0$ to the node is composed of $\binom{w}{ba}$ edges labeled by $ab \mapsto ba$, for all letters $a<b$. Note that a polynomial time algorithm checking whether or not two words are $k$-binomially equivalent has been obtained in \cite{Manea} and is of independent interest.

\begin{Example}
Let us consider the word $w = 1223312$ on the alphabet $\{1,2,3\}$. Its $\sim_2$-equivalence class is $\{ 1223312, 2311223\}$. It can be read from the tree in Figure~\ref{fig:tree}.
Some comments need to be done about this figure. The edges labeled by $12 \mapsto 21$ (resp., $13 \mapsto 31$, $23 \mapsto 32$) are represented in black (resp., red, green). 
In every node, the vertical line separates the longest common prefix (denoted by $p$ in the algorithm) between the word in the node and the word $w'$ from $[w]_{\sim_2}$ we are going to reach. If the current node can be written $pcx$ while the word $w'$ can be written $pdy$, the underlined letter corresponds to the leftmost $d$ occurring in $cx$ (see the algorithm).
Finally, when building the tree, if a path has a number of black (resp. red, green) edges greater than $\binom{w}{21}$ (resp., $\binom{w}{31}$, $\binom{w}{32}$), it is useless to continue computing children of this node, since they won't lead to an element of $[w]_{\sim_2}$. 

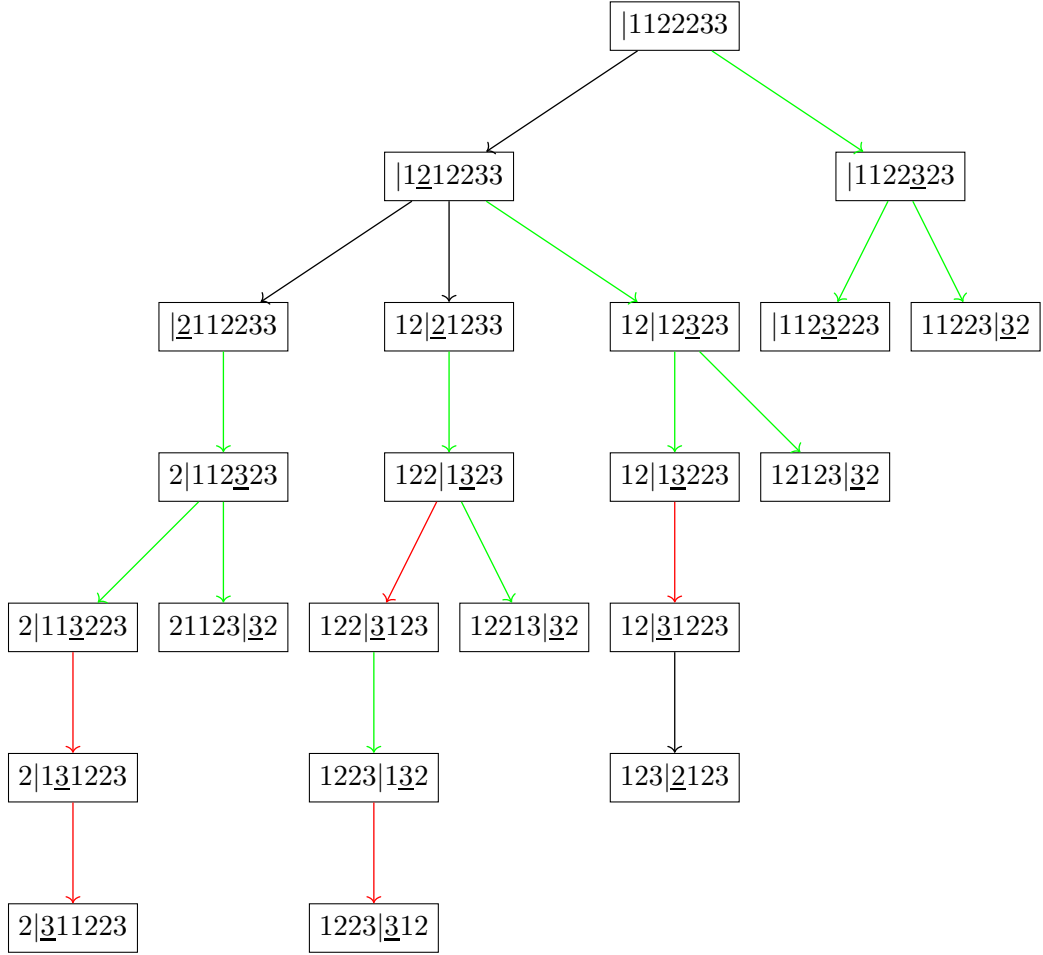
\begin{figure}[h!tbp]
\begin{center}
\begin{tikzpicture}
\tikzstyle{every node}=[shape=rectangle,fill=none,draw=black]
\node (000) at (0,0) {$|1122233$};

\node (100) at (-3,-2) {$|1\underline{2}12233$};
\node (001) at (3,-2) {$|1122\underline{3}23$};

\node (200a) at (-6,-4) {$|\underline{2}112233$};
\node (200b) at (-3,-4) {$12|\underline{2}1233$};
\node (101) at (0,-4) {$12|12\underline{3}23$};
\node (002a) at (2,-4) {$|112\underline{3}223$};
\node (002b) at (4,-4) {$11223|\underline{3}2$};

\node (201a) at (-6,-6) {$2|112\underline{3}23$};
\node (201b) at (-3,-6) {$122|1\underline{3}23$};
\node (102a) at (0,-6) {$12|1\underline{3}223$};
\node (102b) at (2,-6) {$12123|\underline{3}2$};

\node (202a) at (-8,-8) {$2|11\underline{3}223$};
\node (202b) at (-6,-8) {$21123|\underline{3}2$};
\node (211) at (-4,-8) {$122|\underline{3}123$};
\node (202c) at (-2,-8) {$12213|\underline{3}2$};
\node (112) at (0,-8) {$12|\underline{3}1223$};

\node (212a) at (-8,-10) {$2|1\underline{3}1223$};
\node (212b) at (-4,-10) {$1223|1\underline{3}2$};
\node (212c) at (0,-10) {$123|\underline{2}123$};

\node (222a) at (-8,-12) {$2|\underline{3}11223$};
\node (222b) at (-4,-12) {$1223|\underline{3}12$};

\tikzstyle{every node}=[shape=circle,minimum size=5pt,inner sep=2pt]
%fleches correspondant à 01
\tikzstyle{every path}=[color=black, line width = 0.5 pt]
\draw [->] (000) to node [above] {} (100);
\draw [->] (100) to node [above] {} (200a);
\draw [->] (100) to node [above] {} (200b);
\draw [->] (112) to node [above] {} (212c);

%fleches correspondant à 02
\tikzstyle{every path}=[color=red, line width = 0.5 pt]
\draw [->] (201b) to node [above] {} (211);
\draw [->] (102a) to node [above] {} (112);
\draw [->] (202a) to node [above] {} (212a);
\draw [->] (212a) to node [above] {} (222a);
\draw [->] (212b) to node [above] {} (222b);

%fleches correspondant à 12
\tikzstyle{every path}=[color=green, line width = 0.5 pt]
\draw [->] (000) to node [above] {} (001);
\draw [->] (100) to node [above] {} (101);
\draw [->] (001) to node [above] {} (002a);
\draw [->] (001) to node [above] {} (002b);
\draw [->] (200a) to node [above] {} (201a);
\draw [->] (200b) to node [above] {} (201b);
\draw [->] (101) to node [above] {} (102a);
\draw [->] (101) to node [above] {} (102b);
\draw [->] (201a) to node [above] {} (202a);
\draw [->] (201a) to node [above] {} (202b);
\draw [->] (201b) to node [above] {} (202c);
\draw [->] (211) to node [above] {} (212b);

\end{tikzpicture}\caption{Generating the $\sim_2$-class of $1223312$.}\label{fig:tree}
\end{center}
\end{figure}

\end{Example}

\subsection{Isomorphism with a nil-$2$ submonoid}\label{ss32}

Since we are dealing with the extended alphabet $\alphabet\cup\alphabet^{-1}$, let us first introduce a convenient variation of binomial coefficients of words taking into account inverse letters.

\begin{Definition}
Let $t\ge 0$ be an integer. For all words $u$ over the alphabet $\alphabet \cup \alphabet^{-1}$ and $v\in\alphabet^t$, let us define
$$
\begin{bmatrix}
u \\ v
\end{bmatrix} =  
\sum_{(e_1, \ldots, e_{t}) \in \{-1,1\}^{t}} \quad \left( \prod_{i=1}^{t} e_i\right) \quad \binom{u}{v_1^{e_1} \cdots v_{t}^{e_{t}}},
$$
where $\binom{u}{v_1^{e_1} \cdots v_{t}^{e_{t}}}$ is the usual binomial coefficient over the alphabet $\alphabet \cup \alphabet^{-1}$.
Let $w \in \left(\alphabet \cup \alphabet^{-1}\right)^*$ and denote
$$\Phi(w) = \left(
\begin{bmatrix}
w \\ 1
\end{bmatrix}, \ldots,
\begin{bmatrix}
w \\ m
\end{bmatrix},
\begin{bmatrix}
w \\ 12
\end{bmatrix}, \ldots, 
\begin{bmatrix}
w \\ m(m-1)
\end{bmatrix}
\right)^\intercal \in \mathbb{Z}^{m^2}
$$ where the last $m^2- m$ components are obtained from all the words made of two different letters in $\alphabet$, ordered by lexicographical order.

\end{Definition}

Notice that, if $u$ and $v$ are words over $\alphabet$, then
$$\begin{bmatrix}
    u\\ v
\end{bmatrix}=\binom{u}{v}$$ and so, 
$u$ and $v$ are $2$-binomially equivalent if and only if $\Phi(u) = \Phi(v)$.

\begin{Example}
Let $\alphabet = \{1,2,3\}$ and $w = 123^{-1}231^{-1}$.
Applying the previous definition, for all $a \in \alphabet$, we have
$$
\begin{bmatrix}
w \\ a
\end{bmatrix} = \binom{w}{a} - \binom{w}{a^{-1}}.
$$
Similarly, for all $a, b \in \alphabet$, we have
\begin{equation}
    \label{eq:crochet2}
\begin{bmatrix}
w \\ ab
\end{bmatrix} = \binom{w}{ab} - \binom{w}{a^{-1}b} - \binom{w}{ab^{-1}} + \binom{w}{a^{-1}b^{-1}}.
\end{equation}
Therefore, computing classical binomial coefficients, we obtain
\begin{align*}
\Phi(w) =& \left(
0,2,0,2,0,-2,1,0,-1
\right)^\intercal.
\end{align*}
\end{Example}

We are now ready to prove the main result of this section.
\begin{Theorem}\label{thm:isom}
Let $\alphabet=\{1,\ldots,m\}$. The monoid $\alphabet^*/\!\sim_2$ is isomorphic to the submonoid, generated by $\alphabet$, of the nil-$2$ group $N_2(\alphabet)$.
\end{Theorem}
\begin{proof}
Let us first show that for any two words $w$ and $w'$ over $\alphabet \cup \alphabet^{-1}$ such that $\pi(w) = \pi(w')$, the relation $\Phi(w) = \Phi(w')$ holds. Indeed, using~\eqref{eq:crochet2} one can easily check that, for all $a,b\in\alphabet$ and $s,t\in(\alphabet\cup\alphabet^{-1})^*$, we have
$$
\begin{bmatrix}
    st\\
    ab\\
\end{bmatrix}=
\begin{bmatrix}
    s\\
    ab\\
\end{bmatrix}+
\begin{bmatrix}
    t\\
    ab\\
\end{bmatrix}+
\begin{bmatrix}
    s\\ 
    a\\
\end{bmatrix}
\begin{bmatrix}
    t\\
    b\\
\end{bmatrix}.$$
Now, one can show that, for all $u,v\in (\alphabet \cup \alphabet^{-1})^*$ and $x,y,z\in \alphabet \cup \alphabet^{-1}$,
$$\Phi(uv)=\Phi(uxx^{-1}v)\quad \text{ and }\quad \Phi(u[x,y]zv)=\Phi(uz[x,y]v).$$
For instance, let $a,b\in\alphabet$ with $a\neq b$, 
\begin{align*}
\begin{bmatrix}
    uxx^{-1}v\\
    ab\\
\end{bmatrix}&=
\begin{bmatrix}
    u\\ ab\\
\end{bmatrix}+
\begin{bmatrix}
    xx^{-1}v\\ ab\\
\end{bmatrix}+
\begin{bmatrix}
    u\\ a\\
\end{bmatrix}
\begin{bmatrix}
    xx^{-1}v\\ b\\
\end{bmatrix}\\
&= \begin{bmatrix}
    u\\ ab\\
\end{bmatrix}+\underbrace{
\begin{bmatrix}
    xx^{-1}\\ ab\\
\end{bmatrix}}_{=0}+\begin{bmatrix}
    v\\ ab\\
\end{bmatrix}+
\underbrace{\begin{bmatrix}
    xx^{-1}\\ a\\
\end{bmatrix}}_{=0}
  \begin{bmatrix}
      v\\ b\\
  \end{bmatrix}+
  \begin{bmatrix}
      u\\ a\\
  \end{bmatrix}.\left(
  \underbrace{\begin{bmatrix}
       xx^{-1}\\ b\\
  \end{bmatrix}}_{=0}+
  \begin{bmatrix}
      v\\ b\\
  \end{bmatrix}\right)\\
&=
    \begin{bmatrix}
        uv\\ ab\\
    \end{bmatrix}.
\end{align*}
This implies that a map $\Phi_N$ can be defined on the free nil-$2$ group (otherwise stated, the diagram depicted in Figure~\ref{fig:cdiag} is commutative) by 
$$\forall r\in N_2(\alphabet),\quad \Phi_N(r)=\Phi(w) \text{ for any }w\text{ such that }\pi(w)=r.$$
\begin{figure}[h!tbp]
\begin{center}
\begin{tikzpicture}
\tikzstyle{every node}=[fill=none,minimum size=2pt,inner sep=2pt]
\node (alph) at (0,0) {$\alphabet^*$};
\node (alphbis) at (3,0) {$\left( \alphabet \cup \alphabet^{-1} \right)^*$};
\node (nil) at (6,0) {$N_2(\alphabet)$};
\node (z) at (3,-3) {$\mathbb{Z}^{m^2}$};
\draw[->] (alph.south) to node [left] {$\Phi_{| \alphabet^*}$} (z) ;
\draw[->] (alphbis.south) to node [left] {$\Phi$} (z) ;
\draw[->] (nil.south) to node [right] {$\Phi_N$} (z) ;
\draw[->] (alphbis.east) to node [above] {$\pi$} (nil.west);
\path
(alph) edge[{Hooks[right,length=0.4ex]}->] node {} (alphbis); 

\end{tikzpicture}\caption{A commutative diagram (proof of Theorem~\ref{thm:isom}).}\label{fig:cdiag}
\end{center}
\end{figure}
In particular, if $w$ and $w'$ are words over $\alphabet$ such that $\pi(w)=\pi(w')$, then we may conclude that $\Phi(w)=\Phi(w')$ meaning that they are $2$-binomially equivalent. Otherwise stated, for every $r\in N_2(\alphabet)$, $\pi^{-1}(r)\cap \alphabet^*$ is a subset of an equivalence class for $\sim_2$.

To conclude the proof, we have to show that all the elements of an equivalence class for $\sim_2$ are mapped by $\pi$ on the same element of $N_2(\alphabet)$. Let $u,v\in\alphabet^*$ be such that $u\sim_2 v$.
Using the algorithm described in Section~\ref{ss31}, there exists a path in the associated tree from the root $1^{|u|_1}2^{|u|_2}\cdots m^{|u|_m}$ to $u$ and another one to $v$. 
By definition of the commutator, if $u$ is written $pbas$ with $a<b$, then $u = p ab [b,a] s$. Moreover, $\pi(u) = \pi (p ab s [b,a])$ since the commutators are central in $N_2(\alphabet)$. 

Therefore, following backwards the path from $u$ to the root of the tree and recalling that each edge corresponds to an exchange of $2$ letters, we obtain
$$
\pi(u) = \pi \left(1^{|u|_1}2^{|u|_2}\cdots m^{|u|_m} [2,1]^{\binom{u}{21}} \cdots [m,1]^{\binom{u}{m1}}\cdots [m,m-1]^{\binom{u}{m(m-1)}}\right)
$$
and, similarly, following backwards the path from $v$ to the root, 
$$
\pi(v) = \pi \left(1^{|v|_1}2^{|v|_2}\cdots m^{|v|_m} [2,1]^{\binom{v}{21}} \cdots [m,1]^{\binom{v}{m1}}\cdots [m,m-1]^{\binom{v}{m(m-1)}} \right).
$$
But since $u \sim_2 v$, we get $\pi(u) = \pi(v)$. 
\qedhere
\end{proof}

\section{Growth order}\label{sec:4}

We first show that the growth of $\# (\alphabet^n /\!\sim_k)$ is bounded by a polynomial in $n$. This generalizes a result from \cite{RigoSalimov:2015} for a binary alphabet. Note that a similar result was obtained in \cite{Lejeune}. Next, we obtain better estimates for $\sim_2$. 

\begin{Proposition}\label{prop:growth}
Let $\alphabet = \{1, \ldots, m\}$ and $k \geq 1$. We have
$$
\# (\alphabet^n /\!\sim_k) \in \mathcal{O} \left(n^{k^2 m^k}\right)
$$
when $n$ tends to infinity.
\end{Proposition}
\begin{proof}
For every $u,v \in \alphabet^*$ such that $1 \leq |v| \leq k$ and $|u| = n$, we have
\begin{align*}
0 \leq \binom{u}{v} \leq \binom{|u|}{|v|} \leq n^{|v|} \leq n^k.
\end{align*}
Therefore, for every $v$ such that $1 \leq |v| \leq k$, we have
$$
\# \left\{ \binom{u}{v} \,:\, |u|=n \right\} \leq n^k +1.
$$
By definition, the $\sim_k$-equivalence class of $u$ is uniquely determined by the values of 
$ \binom{u}{v}$ for all $v\in \alphabet^*$ such that $1 \leq |v| \leq k$.
There are $$
\sum_{i=1}^k m^i \leq k m^k
$$ such coefficients and thus,
$$
\# (\alphabet^n /\!\sim_k) \leq \left(n^k+1\right)^{km^k}.
$$
\qedhere
\end{proof}

We have obtained an upper bound which is far from being optimal but it ensures that the growth is polynomial. However, for $k=2$, it is possible to obtain the polynomial degree of the growth. We make use of Landau notation: $f \in \Theta(g)$ if there exist constants $A, B >0$ such that, for all $n$ large enough, $A\, g(n) \leq f(n) \leq B\, g(n)$.

\begin{Proposition}
Let $\alphabet= \{1,\ldots, m\}$ be an alphabet of size $m \geq 2$. We have 
$$ \# (\alphabet^n /\!\sim_2) \in \Theta \left(n^{m^2 - 1}\right)$$
when $n$ tends to infinity.
\end{Proposition}
\begin{proof}
Let $f$ be the function such that for any $x\in\mathbb{N}^m$, $$f(x)=\# (\{u\in\alphabet^*: \Psi(u)=x\} /\!\sim_2).$$ 
In other words, $f(x)$ counts the number of $2$-binomial equivalence classes whose Parikh vector is $x$.
Let $||\!\cdot\!||_1:\mathbb{R}^d\rightarrow\mathbb{R}$ be the $\ell_1$-norm (i.e., for all vectors $v$, $||v||_1=\sum_{i=1}^d|v_i|$).
Clearly for all $n$, 
\begin{equation}\label{functionoff}
 \# (\alphabet^n /\!\sim_2) =\sum_{x\in\mathbb{N}^m, ||x||_1=n}f(x).
\end{equation}

For any $a,b\in\alphabet$, $a<b$, and $u\in\alphabet^*$, $\binom{u}{ba}= |u|_a|u|_b-\binom{u}{ab}$ and  $\binom{u}{aa}=\binom{|u|_a}{2}$. 
Any word $u$ has its equivalence class uniquely determined by  the values of $|u|_a$ for all $a\in\alphabet$ and $\binom{u}{ab}$ for all $a<b\in\alphabet$.
Moreover, for all $u\in\alphabet^*$ and $a<b\in\alphabet$, $\binom{u}{ab}\le|u|_a|u|_b$.
We deduce that for all $x=(x_1,\ldots,x_m)\in\mathbb{N}^m$, $$f(x)\le\prod_{1\le a<b \le m}x_ax_b\le \prod_{1\le a<b\le m}||x||_1^2
\le ||x||_1^{m(m-1)}.$$
From equation \eqref{functionoff}, we get that
\begin{align*}
  \# (\alphabet^n /\!\sim_2) &\le\sum_{x\in\mathbb{N}^m, ||x||_1=n}||x||_1^{m(m-1)}\\
  \# (\alphabet^n /\!\sim_2) &\le n^{m(m-1)} \;\; \# \left\{  x\in\mathbb{N}^m \, : \, ||x||_1=n \right\} \\
  \# (\alphabet^n /\!\sim_2) &\le n^{m(m-1)}  (n+1)^{m-1}  \le (n+1)^{m^2-1}.
\end{align*}
We conclude that $\# (\alphabet^n /\!\sim_2) \in \mathcal{O} \left(n^{m^2-1} \right)$ when $n \rightarrow + \infty$. It remains to get a convenient lower bound. 

For any $a,b\in\alphabet$, $a\neq b$, and $i,j\in\mathbb{N}$, let 
$$L_{a,b,i,j}=\{u\in\{a,b\}^*:|u|_a=i, |u|_b=j\}.$$
Considering all possible letter exchanges as in \eqref{eq:moinsUn} from $a^ib^j$ to $b^ja^i$, 
the binomial coefficient $\binom{u}{ab}$ decreases by $1$ at every step from $ij$ to $0$, we thus have 
\begin{align}\label{eq:card}
\left\{\binom{u}{ab}:u\in L_{a,b,i,j}\right\} =\left\{ 0, 1, \ldots, ij\right\}
\end{align}
which is a set of cardinality $ij+1$.
% === a partir d'ici division par m-1 et pas m ! === 
For any $x\in\mathbb{N}^m$, let us consider the following language 
$$L(x)=\left(\prod_{a=1}^m \prod_{b=a+1}^m L_{a,b,\left\lfloor\frac{x_a}{m-1}\right\rfloor,\left\lfloor\frac{x_b}{m-1}\right\rfloor}\right)
\prod_{a=m}^{1} a^{x_a \% (m-1)},$$
where the products must be understood as languages concatenations, the indices of the last product are taken in decreasing order, and $x \% y$ is the remainder of the Euclidean division of $x$ by $y$.

For instance for $m=3$, 
$$L(x)=\left\{u_{1,2}u_{1,3}u_{2,3}r_3r_2r_1:\forall a,b\in \alphabet,  u_{a,b}\in L_{a,b,\left\lfloor\frac{x_a}{2}\right\rfloor,\left\lfloor\frac{x_b}{2}\right\rfloor}, r_a=a^{x_a \% 2}\right\}.$$

%$$L(x)=
%L_{0,1,\left\lfloor\frac{x_0}{m}\right\rfloor,\left\lfloor\frac{x_1}{m}\right\rfloor}
%L_{0,2,\left\lfloor\frac{x_0}{m}\right\rfloor,\left\lfloor\frac{x_2}{m}\right\rfloor}
%L_{1,2,\left\lfloor\frac{x_1}{m}\right\rfloor,\left\lfloor\frac{x_2}{m}\right\rfloor}
%0^{x_0-m\left\lfloor\frac{x_i}{m}\right\rfloor} i^{x_1-m\left\lfloor\frac{x_1}{m}\right\rfloor} i^{x_2-m\left\lfloor\frac{x_2}{m}\right\rfloor}$$

Roughly speaking, for every $a<b$ and $u\in L(x)$, we will show that $\binom{u}{ab}$ mostly depends on $u_{a,b}$ and takes a quadratic number of values (when choosing $u_{a,b}$ accordingly). Furthermore, the role of $r_a$ words is limited to padding. Indeed, observe that for all $u\in L(x)$, $\Psi(u)=x$.

Let $x\in\mathbb{N}^m$ and $u\in L(x)$. Then, by definition, there exist words $u_{1,2},u_{1,3},\ldots, u_{m-1,m}$ with, for all $a<b$, $u_{a,b}\in L_{a,b,\left\lfloor\frac{x_a}{m-1}\right\rfloor,\left\lfloor\frac{x_b}{m-1}\right\rfloor}$, such that
$$u=\left(\prod_{a=1}^m \prod_{b=a+1}^m u_{a,b}\right)\prod_{a=m}^1 a^{x_a \% (m-1)}.$$
Let $i$ and $j$ be two integers such that $1\le i<j \le m$ and let us compute the binomial coefficient associated with $ij$. A subword $ij$ either occurs in a single factor of the above product (the first two terms below), or $i$ and $j$ appear in two different factors:
\begin{align*}
\binom{u}{ij}&=\sum_{a=1}^m \sum_{b=a+1}^m \binom{u_{a,b}}{ij}+\sum_{a=m}^1 \binom{a^{x_a \% (m-1)}}{ij}\\
&+\sum_{a<b\in\alphabet}\left|u_{a,b}\right|_i\left(\sum_{\substack{a'<b'\in\alphabet\\(a',b')>(a,b)}} \left|u_{a',b'}\right|_j+\sum_{b\in\alphabet} \left|b^{x_b \% (m-1)}\right|_j\right)\\
&+\sum_{a=1}^m \sum_{b=1}^{a-1}\left|a^{x_a \% (m-1)}\right|_i\left|b^{x_b \% (m-1)}\right|_j.
\end{align*}
Observe that by definition of $L(x)$, the second and last terms vanish. Hence, 
\begin{align*}
\binom{u}{ij}&= \binom{u_{i,j}}{ij}
+ \underbrace{   \sum_{\substack{a<b\in\alphabet \\ a=i \text{ or }  b=i}}\left\lfloor\frac{x_i}{m-1}\right\rfloor
\left( \left(\sum_{\substack{a'<b' \in \alphabet \\ (a',b') > (a,b) \\ a'=j \text{ or }  b'=j}} \left\lfloor\frac{x_j}{m-1}\right\rfloor \right) +x_j \%(m-1)\right) }_{:=h_{i,j}(x)}
\end{align*}
The second term of the latter expression is uniquely a function of $x$ (there is no dependency on $u$) while, from \eqref{eq:card}, 
$$\left\{\binom{u_{i,j}}{ij}:u_{i,j}\in L_{i,j, \left\lfloor\frac{x_i}{m-1}\right\rfloor, \left\lfloor\frac{x_j}{m-1}\right\rfloor}\right\} =\left\{0,1,\ldots,
\left\lfloor\frac{x_i}{m-1}\right\rfloor\left\lfloor\frac{x_j}{m-1}\right\rfloor+1\right\}.$$
Thus for a fixed $x$, considering all $u \in L(x)$, $\binom{u}{ij}$ can take $\left\lfloor\frac{x_i}{m-1}\right\rfloor\left\lfloor\frac{x_j}{m_1}\right\rfloor+1$ different values. Moreover, for all $\left(a_{1,2}, a_{1,3}, \ldots, a_{(m-1),m} \right)$ such that 
$$
a_{i,j} \in \left\{ h_{i,j}(x), \ldots, h_{i,j}(x) + \left\lfloor\frac{x_i}{m-1}\right\rfloor\left\lfloor\frac{x_j}{m-1}\right\rfloor \right\} \qquad \forall i < j,
$$
there exists $u \in L(x)$ such that $\binom{u}{ij} = a_{i,j}$ for all $i<j$. We deduce that, for all~$x$,
$$f(x)\ge\prod_{a<b\in\alphabet}\left( \left\lfloor\frac{x_a}{m-1}\right\rfloor\left\lfloor\frac{x_b}{m-1}\right\rfloor+1 \right).$$
By equation \eqref{functionoff}, we finally get the lower bound:
\begin{align*}
  \# (\alphabet^n /\!\sim_2) &\ge\sum_{\substack{x\in\mathbb{N}^m \\ ||x||_1=n}}\prod_{a<b\in\alphabet}\left( \left\lfloor\frac{x_a}{m-1}\right\rfloor\left\lfloor\frac{x_b}{m-1}\right\rfloor+1 \right)\\
&\ge\sum_{\substack{x\in\mathbb{N}^m\\ ||x||_1=n\\ \forall i, x_i\ge\frac{n}{2m}+m}}\prod_{a<b\in\alphabet}  \left\lfloor\frac{x_a}{m-1}\right\rfloor\left\lfloor\frac{x_b}{m-1}\right\rfloor\\
&\ge\sum_{\substack{x\in\mathbb{N}^m\\ ||x||_1=n\\ \forall i, x_i\ge\frac{n}{2m}+m}}\prod_{a<b\in\alphabet}\left(\frac{n}{2m(m-1)}\right)^2\\
&\ge\left(\frac{n}{2m(m-1)}\right)^{m(m-1)} \; \# \left\{ x\in\mathbb{N}^m : ||x||_1=n \ \wedge \  \forall i, x_i\ge\frac{n}{2m}+m   \right\}.
\end{align*}
The latter set contains the set
$$
\left\{ x\in\mathbb{N}^m : ||x||_1=n  \quad \wedge \quad  \forall i, x_i\ge\frac{n}{2m}+m \quad  \wedge \quad  \forall i < m, x_i \le \frac{n}{m}   \right\}.
$$
For $n$ large enough (i.e., $\frac{n}{2m} + m \leq \frac{n}{m}$), the cardinal of this set is
$$
\left( \frac{n}{2m} - m +1 \right)^{m-1} \in \Theta (n^{m-1}).
$$
Moreover,
\begin{align*}
\left( \frac{n}{2m(m-1)} \right) ^{m(m-1)} \in \Theta\left(n^{m(m-1)}\right)
\end{align*}
and we conclude that
$\# (\alphabet^n /\!\sim_2) \in \Theta \left(n^{m^2 - 1}\right)\,.$
\end{proof}

\begin{Remark}
Note that even though the growth of $\# (\{1,2,3\}^n /\!\sim_2)$ is polynomial, this quantity is not a polynomial. It is easy to verify by interpolating the $9$ first values.
A similar remark can be obtained for $\# (\{1,2\}^n /\!\sim_3)$ whose first values can be found as entry \texttt{A258585} in Sloane's encyclopedia.
\end{Remark}

\section{Non context-freeness}\label{sec:5}
In this section, we show that for any alphabet $\alphabet$ of size at least $3$ and for any $k \geq 2$, the languages  $\LL(\sim_k,\alphabet)$ and $\sing( \sim_k,\alphabet)$ are not context-free.

Let $L\subseteq\alphabet^*$ be a language. The {\em growth function} of $L$ maps the integer $n$ to $\# (L\cap\alphabet^n)$. A language has a {\em polynomial growth} if there exists a polynomial $p$ such that $\# (L\cap\alphabet^n)\le p(n)$ for all $n\ge 0$. 
%A language has an {\em exponential growth} if there exists a real number $r>1$ such that $\# (L\cap\alphabet^n)\ge r^n$ for infinitely many $n$. 
Recall that a language $L$ is {\em bounded} if there exist words $w_1,\ldots,w_\ell\in\alphabet^*$ such that $L\subseteq w_1^*w_2^*\cdots w_\ell^*$. Ginsburg and Spanier have obtained many results about bounded context-free languages, see \cite{GS}. We will make use of the following result. For relevant bibliographic pointers see, for instance, \cite{Gaw}. 
\begin{Proposition}\label{prop:nonContext}
A context-free language is bounded if and only if it has a polynomial growth.    
\end{Proposition}

We easily deduce\footnote{$\sing(\sim_k,\alphabet)\subseteq\LL(\sim_k,\alphabet)$ and $\LL(\sim_k,\alphabet)$ is in one-to-one correspondence with $\alphabet^*\!/\!\sim_k$.} from the previous section that both languages $\LL(\sim_k,\alphabet)$ and $\sing( \sim_k,\alphabet)$ have polynomial growth; it is thus enough to show that they are not bounded to infer that they are not context-free. Observe that, in our forthcoming reasonings, we will define particular words $\rho_{p,n}$ over a ternary alphabet (they can trivially be seen as words over a larger alphabet). 
%Indeed, if $L$ is a context-free language over an arbitrary alphabet $\{1,\ldots,m\}^*$, then $L\cap\{1,2,3\}^*$ must be context-free (intersection of a context-free language and a regular language). 

\begin{Definition}
Fix a sequence $(s_n)_{n\ge 1}$ of positive integers such that, for all $n \in \N$,
 \begin{align*}
 &\sqrt{\frac{s_{n}}{2}} \in \mathbb{N},\tag{D1}\label{d2}\\
  &s_{n}>\left(\sqrt{\frac{s_{n}}{2}}+\sum\limits_{i=1}^{n-1}s_i\right)^2,\label{d1}\tag{D2}\\
  &\sqrt{\frac{s_{n}}{2}} >\left(\sum_{i=1}^{n-1}s_{i}\right)\left(\sum_{i=1}^{n-3}s_{i}\right).\tag{D3}\label{d3}
\end{align*}
%\\
%  &s_n>\left(\sum_{i=1}^{n-1}s_i \right)\left(\sum_{i=1}^{n-3}s_i \right).\tag{D4}\label{d4}
For instance, to get a sequence with those prescribed properties, one can choose $$s_n=2\times 8^{8^n}.$$
For any integers $n$ and $p$, let us define the word $$\rho_{p,n}=1^{p}2^{s_{n-1}}3^{s_{n-2}}1^{s_{n-3}}\cdots a^{s_1}$$
over $\{1,2,3\}$, where $a \equiv n \pmod 3$.
\end{Definition}

In Section~\ref{subs:sing}, we will prove that the $\sim_2$-class of any $\rho_{p,n}$ is a singleton. Then, it is proven in Section~\ref{subs:notBounded} that $\{ \rho_{p,n} \mid p,n \in \N \}$ is not a bounded language. Putting together these results, we get the following.

\begin{Theorem}
For any alphabet $\alphabet$ of size at least $3$ and for any $k \geq 2$, the languages $\LL(\sim_k, \alphabet)$ and $\sing(\sim_k, \alphabet)$ are not context-free.
\end{Theorem}

\begin{proof}
First note that $$
\{ \rho_{p,n} \mid p,n \in \N \} \subseteq \{1,2,3\}^* \subseteq \alphabet^*.$$
Taking into account Corollary~\ref{cor:difu}, observe that
$$
\{ \rho_{p,n} \mid p,n \in \N \} \subseteq \sing(\sim_k, \alphabet) \subseteq \LL(\sim_k, \alphabet).
$$

From Proposition~\ref{prop:growth}, the languages $\LL(\sim_k, \alphabet)$, and thus $\sing(\sim_k, \alphabet)$, have a polynomial growth.
From Lemma~\ref{lem:notBounded}, the language $\{ \rho_{p,n} \mid p,n \in \N \}$ is not bounded.
Therefore, $\sing(\sim_k, \alphabet)$ and $\LL(\sim_k, \alphabet)$ are not bounded and we conclude from Proposition~\ref{prop:nonContext}.
\end{proof}

\begin{Remark}
This result is in fact true for all languages having exactly one representant of each $\sim_k$-class.
\end{Remark}

\subsection{A family of singletons}\label{subs:sing}

\begin{Proposition}\label{mainlemma}
For any two positive integers $n$ and $p$ and word $u$, at least one of the following is false:
\begin{itemize}
  \item $u\not=\rho_{p,n}$,
  \item $\Psi(u)=\Psi(\rho_{p,n})$,
  \item $\binom{u}{12}\ge\binom{\rho_{p,n}}{12}$,
  \item $\binom{u}{23}\ge\binom{\rho_{p,n}}{23}$,
  \item $\binom{u}{31}\ge\binom{\rho_{p,n}}{31}$.
\end{itemize}
\end{Proposition}
As an immediate corollary, we get the following result.
\begin{Corollary}\label{cor:difu}
For any two positive integers $n$ and $p$ and word $u$ such that
 $u\not=\rho_{p,n}$, we have $u\not\sim_2\rho_{p,n}$.
\end{Corollary}

\begin{proof}[Proof of Proposition~\ref{mainlemma}]
Let us show the proposition by induction on $n$.
The result clearly holds for $n\le3$. Let $n\ge4$ be an integer such that the result holds for any $i<n$.
Now let us proceed by contradiction to show that the result also holds for $n$.

For the sake of contradiction, let $p$ and $u$ be such that
\begin{align}
 u&\not=\rho_{p,n},\label{eqdiff}\\
  \Psi(u)&=\Psi(\rho_{p,n}),\label{eqabel}\\
  \binom{u}{12}&\ge\binom{\rho_{p,n}}{12},\label{eq01}\\
  \binom{u}{23}&\ge\binom{\rho_{p,n}}{23},\label{eq12}\\
  \binom{u}{31}&\ge\binom{\rho_{p,n}}{31}.\label{eq20}
\end{align}
Let $u=vw$ where $v$ is the prefix of length $p+s_{n-1}-\sqrt{\frac{s_{n-1}}{2}}$ of $u$. Similarly, let $\rho_{p,n}=v'w'$ where $|v| = |v'|$.
In the first part of the proof, we show that $\Psi(v) = \Psi(v')$ and more precisely $|v|_1= p = |v'|_1$, $|v|_3=0=|v'|_3$. We proceed into three steps.

\smallskip

$\bullet$ Proof of $|v|_1\ge p$:

For the sake of contradiction, suppose $|v|_1\le p-1$. Then
\begin{align*}
\binom{u}{12}&=\binom{v}{12}+\binom{w}{12}+ |v|_1|w|_2 \\
&\le |v|_1|v|_2+|w|_1|w|_2+|v|_1|w|_2 \\
&\le |v|_1|u|_2+|w|_1|w|_2 \\
&\le (p-1)|u|_2+ |w|^2.
\end{align*}
Replacing $|w|$ by its value, we get
\begin{align}
\binom{u}{12} &\le p|u|_2+ \left(\sqrt{\frac{s_{n-1}}{2}}+\sum_{i=1}^{n-2}s_i\right)^2-|u|_2\label{eqinter1}.
\end{align}
By \eqref{eqabel}, $|u|_2=|\rho_{p,n}|_2$ and condition \eqref{d1} implies
$$0>\left(\sqrt{\frac{s_{n-1}}{2}}+\sum_{i=1}^{n-2}s_i\right)^2-s_{n-1}\ge\left(\sqrt{\frac{s_{n-1}}{2}}+\sum_{i=1}^{n-2}s_i\right)^2-|u|_2.$$
 Together with \eqref{eqinter1}, it gives $\binom{u}{12}< p|\rho_{p,n}|_2\leq \binom{\rho_{p,n}}{12}$.
This is a contradiction with hypothesis \eqref{eq01} and we conclude that $|v|_1\ge p$.

\smallskip

$\bullet$ Proof of $|v|_3=0$:

For the sake of contradiction, suppose $|v|_3\ge 1$. 
\begin{align}
\binom{u}{23}&= |u|_2|u|_3 - \binom{u}{32}\nonumber \\ 
&= |u|_2|u|_3-\binom{v}{32} - \binom{w}{32} - |v|_3 |w|_2 \nonumber \\
&\le |u|_2|u|_3-|v|_3|w|_2\nonumber \\
&\le |u|_2|u|_3-|w|_2. \label{eq:23}
\end{align}
Observe that
\begin{align*} 
 |w|_2 &= |u|_2-|v|_2 \\
 &=|u|_2-|v|+|v|_1+|v|_3 \\
 &> |u|_2-|v|+|v|_1
 \end{align*}
 and 
 \begin{align*}
|u|_2-|v|+|v|_1 
 &\geq |u|_2 - p - s_{n-1} + \sqrt{\frac{s_{n-1}}{2}} + |v|_1  \\
 &\geq (|u|_2 - s_{n-1}) + (|v|_1 - p) + \sqrt{\frac{s_{n-1}}{2}}  \\
 &\geq \sqrt{\frac{s_{n-1}}{2}}\,.
 \end{align*}
 Moreover, by \eqref{eqabel}, $|u|_2|u|_3=|\rho_{p,n}|_2|\rho_{p,n}|_3$. We can use these two remarks in inequality~\eqref{eq:23}.
\begin{align*}
\binom{u}{23}&< |\rho_{p,n}|_2|\rho_{p,n}|_3-\sqrt{\frac{s_{n-1}}{2}}\\
  &<|\rho_{p,n}|_2 |\rho_{p,n}|_3 -\left(\sum_{i=1}^{n-2}s_{i}\right)\left(\sum_{i=1}^{n-4}s_{i}\right) \,\,\text{(from \eqref{d3})}\\
 &< |\rho_{p,n}|_2 |\rho_{p,n}|_3 -\sum_{i=0}^{\left\lfloor \frac{n-3}{3}\right\rfloor}\left( s_{n-2-3i} \sum_{j=i}^{\left\lfloor \frac{n-5}{3}  \right\rfloor}s_{n-4-3j}\right)
 \end{align*}
The latter quantity is equal to $|\rho_{p,n}|_2 |\rho_{p,n}|_3-\binom{\rho_{p,n}}{32}$ and thus 
$$
\binom{u}{23} < \binom{\rho_{p,n}}{23}.
$$
This contradicts hypothesis \eqref{eq12} and we conclude that $|v|_3=0$.

\smallskip

$\bullet$ Proof of $|v|_1\le p$:

For the sake of contradiction, suppose $|v|_1\ge p+1$. Then
$$\binom{u}{13}\ge |v|_1|w|_3\ge p|w|_3+ |w|_3.$$
Since $|v|_3=0$, $|w|_3=|u|_3=|\rho_{p,n}|_3\ge s_{n-2} > \sqrt{\frac{s_{n-2}}{2}}$.
From condition \eqref{d3}, taking into account the structure of $\rho_{p,n}$, we deduce
$$\binom{u}{13}> p|\rho_{p,n}|_3 +\left(\sum_{i=1}^{n-3}s_i \right)\left(\sum_{i=1}^{n-5}s_i \right)\ge\binom{\rho_{p,n}}{13}.$$
This yields a contradiction with hypothesis \eqref{eq20}. We conclude that $|v|_1= p$. We thus have 
\begin{equation}
    \label{eq:psi13}
\Psi(v')=\Psi(v) \text{ and }\Psi(w')=\Psi(w).    
\end{equation}

We will now use it to find the contradiction. From hypothesis \eqref{eq01}, we get
\begin{align}
	\binom{w}{12} &= \binom{u}{12} - \binom{v}{12} - |v|_1 |w|_2 \nonumber \\
&\ge\binom{\rho_{p,n}}{12} - \binom{v}{12} - |v|_1 |w|_2 \nonumber \\
&\ge\binom{w'}{12}+\binom{v'}{12}-\binom{v}{12}+ |v'|_1|w'|_2-|v|_1|w|_2 \nonumber \\
&\ge\binom{w'}{12}+\binom{v'}{12}-\binom{v}{12} \label{eq:discv}
\end{align}
where the last inequality is due to \eqref{eq:psi13}.

Since $v'$ is of the form $1^\alpha 2^\beta$,
$$
\binom{v'}{12} = \max \left\{ \binom{x}{12} \,: \, \Psi(x) = \Psi(v') \right\} \ge \binom{v}{12}
$$
and we get $\binom{w}{12} \geq \binom{w'}{12}$. With similar arguments and the fact that $|v|_3 = 0 = |v'|_3$, we obtain 
$\binom{w}{23} \geq \binom{w'}{23}$ and $\binom{w}{31} \ge \binom{w'}{31}$.

Observe that $w \neq w'$. Indeed, it is obvious if $v = v'$ (since $vw \neq v'w'$) and otherwise, $\binom{v'}{12} > \binom{v}{12}$ and from \eqref{eq:discv}, we get $\binom{w}{12} > \binom{w'}{12}$.

Let $\sigma$ be the morphism such that $\sigma(1)=3;\sigma(2)=1;\sigma(3)=2.$ Then $\sigma(w')=\rho_{ \sqrt{\frac{s_{n-1}}{2}},n-1}$ and since $\sigma$ is a permutation of the alphabet, we get
 \begin{itemize}
  \item $\sigma(w) \not= \sigma(w')$,
  \item $\Psi(\sigma(w))=\Psi(\sigma(w'))$,
  \item  $\binom{\sigma(w)}{12}\ge\binom{\sigma(w')}{12}$,
  \item  $\binom{\sigma(w)}{23}\ge\binom{\sigma(w')}{12}$,
  \item   $\binom{\sigma(w)}{31}\ge\binom{\sigma(w')}{31}$.
 \end{itemize}
This is a contradiction with our induction hypothesis.
We deduce that there is no such pair of integers and this concludes the proof of the proposition.
\end{proof}

\subsection{Unboundedness}\label{subs:notBounded}
It remains us to prove that the language $\{\rho_{p,n} : p,n \in \N\}$ is not bounded. We will make use of the following notation. For all non-empty words $w \in \alphabet^+$, its \textit{letter-factorization} is $(c_1, q_1), \ldots, (c_{r},q_{r})$, where
$$
w = c_1^{q_1} c_2^{q_2} \cdots c_{r}^{q_{r}},
$$
$r \geq 1$, $c_1, \ldots, c_{r}$ are letters such that for all $i$, $c_i \neq c_{i+1}$, and where $q_1, \ldots, q_{r}$ are positive integers.  The \textit{number of blocks} in the word $w$, denoted by $nb(w)$, is $r$. It corresponds to the length of the decomposition.

\begin{Example}
Let $w = 112333122132$. We have $c_1 = 1$, $c_2 = 2$, $c_3 = 3$, $c_4 = 1$, $c_5 = 2$, $c_6 = 1$, $c_7 = 3$, $c_8 = 2$, and $q_1 = 2$, $q_2 = 1$, $q_3 = 3$, $q_4 = 1$, $q_5 = 2$, $q_6=q_7=q_8 =1$. Moreover, $nb(w) = 8$.
 \end{Example}
 
The letter-factorization of a word of the form $\rho_{p,n}$ has particular properties that we record in the following remark.

\begin{Remark}\label{rem:rho}
For all $p,n \in \N$, if $(c_1, q_1), \ldots, (c_{r},q_{r})$ is the letter-factorization of $\rho_{p,n}$, we know that
\begin{itemize}
\item[$\bullet$]
for all $i \geq 1$, $c_i \equiv i \pmod{3}$, with $c_i \in \{1,2,3\}$;
\item[$\bullet$]
$q_1 = p $ and for all $i > 1$, $q_i = s_{n-i+1}$;
\item[$\bullet$]
$nb(\rho_{p,n}) = n$.
\end{itemize}
\end{Remark}

\begin{Lemma}\label{lem:notBounded}
For all $\ell \in \N$ and words $w_1, \ldots, w_\ell \in \alphabet^*$, we have
$$
\{ \rho_{p,n} : p,n \in \N \} \not\subset w_1^* \cdots w_\ell^*.
$$
\end{Lemma}

\begin{proof}
For the sake of contradiction, let us assume that there exist $\ell \in \N$ and words $w_1, \ldots, w_\ell \in \alphabet^*$ such that
$$
\mathcal{R} := \{ 1^p 2^{s_{n-1}} 3^{s_{n-2}} \cdots : p \in \N, n \in \N \} \subseteq w_1^* \cdots w_\ell^*.
$$
We will first show that, under this assumption, there exist $N \in \N$ and words $z_1,\ldots,z_q$ such that, for all $i$, $nb(z_i) \leq 2$, and the subset of $\mathcal{R}$
$$
\mathcal{R}_{N} := \{ \rho_{p,n} : p \in \N, n \geq N \}
$$
 is included in $z_1^* \ldots z_{q}^*$.

Let us take the least $i \in \{1, \ldots, \ell \}$ such that $nb(w_i) \geq 3$. If such an $i$ does not exist, we can take $N = 0$, $\ell=q$ and $z_i=w_i$ for all $i$.  
Otherwise, the letter-factorization of $w_i$ begins with $(a_1, \alpha_1), (a_2, \alpha_2), (a_3, \alpha_3)$. Assume that there exist $p,n$ such that the factorization of $\rho_{p,n}$ in terms of $w_1, \ldots, w_\ell$ 
$$\rho_{p,n}=w_1^{n_1}\cdots w_i^{n_i}\cdots w_\ell^{n_\ell}$$ 
contains an occurrence of $w_i$, i.e., $n_i>0$, (if this is not the case, $\mathcal{R}$ is thus included in $w_1^*\cdots w_{i-1}^* w_{i+1}^*\cdots w_\ell^*$ and we can proceed to the next index such that $nb(w_i) \geq 3$).
Because of Remark~\ref{rem:rho}, if $nb(w_j) = 2$ then $w_jw_j$ is never a factor of a word in $\mathcal{R}$ (this would mean that two letters out of three are alternating). In that case, we must have $n_j=1$ in the above factorization. Also, if $nb(w_j)=1$, then $nb(w_{j}^{n_j})=1$. By definition of $i$, $nb(w_j)\le 2$ for all $j<i$. 
Therefore there exists $\gamma \leq 2i$ such that, if $(c_1, q_1), \ldots, (c_r,q_r)$ is the letter-factorization of $\rho_{p,n}$, 
\begin{itemize}
\item[$\bullet$]
$c_1^{q_1} \cdots c_{\gamma - 1}^{q_{\gamma - 1}} \in w_1^* \cdots w_{i-1}^*  a_1^{\alpha_1}$,
\item[$\bullet$]
$c_\gamma = a_2$ and $q_{\gamma} = \alpha_2$,
\item[$\bullet$]
$c_{\gamma +1} = a_3$.
\end{itemize}
See Figure~\ref{fig:rho} for an illustration.

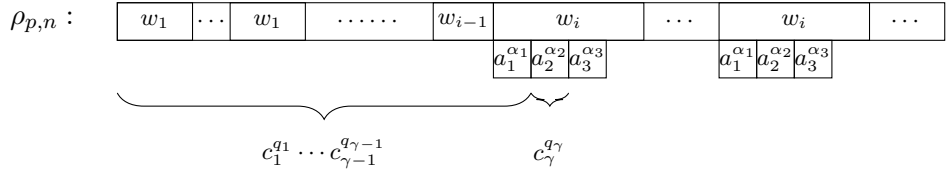
\begin{figure}[h!]
\begin{center}
\begin{tikzpicture}
\draw (0,0) rectangle (11,0.5);
\draw (0,0) rectangle (1,0.5);
\draw (1.5,0) rectangle (2.5,0.5);
\node (0) at (0.5,0.25) {\footnotesize{$w_1$}};
\node (1) at (1.275,0.25) {\footnotesize{$\cdots$}};
\node (2) at (2,0.25) {\footnotesize{$w_1$}};
\node (4) at (3.35,0.25) {\footnotesize{$\cdots\cdots$}};
\draw (4.2,0) rectangle (5,0.5);
\draw (5,0) rectangle (7,0.5);
\draw (8,0) rectangle (10,0.5);
\node (31) at (4.6,0.25) {\footnotesize{$w_{i-1}$}};
\node (4) at (6,0.25) {\footnotesize{$w_i$}};
\node (5) at (7.5,0.25) {\footnotesize{$\cdots$}};
\node (6) at (9,0.25) {\footnotesize{$w_i$}};
\node (7) at (10.5,0.25) {\footnotesize{$\cdots$}};

\draw (5,0) rectangle (5.5,-0.5);
\draw (5.5,0) rectangle (6,-0.5);
\draw (6,0) rectangle (6.5,-0.5);
\draw (8,0) rectangle (8.5,-0.5);
\draw (8.5,0) rectangle (9,-0.5);
\draw (9,0) rectangle (9.5,-0.5);

\node (8) at (5.25,-0.25) {\footnotesize{$a_1^{\alpha_1}$}};
\node (9) at (5.75,-0.25) {\footnotesize{$a_2^{\alpha_2}$}};
\node (10) at (6.25,-0.25) {\footnotesize{$a_3^{\alpha_3}$}};
\node (11) at (8.25,-0.25) {\footnotesize{$a_1^{\alpha_1}$}};
\node (12) at (8.75,-0.25) {\footnotesize{$a_2^{\alpha_2}$}};
\node (13) at (9.25,-0.25) {\footnotesize{$a_3^{\alpha_3}$}};

\draw [decorate,decoration={brace,amplitude=10pt,mirror}]
(0,-0.7) -- (5.5,-0.7);
\node (14) at (2.75,-1.5) {\footnotesize{$c_1^{q_1} \cdots c_{\gamma - 1}^{q_{\gamma-1}}$}};

\draw [decorate,decoration={brace,amplitude=5pt,mirror}]
(5.5,-0.7) -- (6,-0.7);

\node (15) at (5.75,-1.5) {\footnotesize{$c_{\gamma}^{q_{\gamma}}$}};

\node (16) at (-1,0.25) {$\rho_{p,n} : $};

\end{tikzpicture}
\end{center}
\caption{Decomposition of $\rho_{p,n}$ into blocks}
\label{fig:rho}
\end{figure}

With Remark~\ref{rem:rho}, we know that if $\gamma = 2$ then $q_{\gamma-1}=p$ and in all cases, $q_\gamma = s_{n-\gamma+1}$. Therefore, if we take $N$ such that $s_{N-2i+1} > \alpha_2$, the set 
$\mathcal{R}_{N}$, which is included in $w_1^* \ldots w_\ell^*$ is also included in $w_1^* \ldots w_{i-1}^* w_{i+1}^* \ldots w_\ell^*$. We can proceed the same way to eliminate other factors $w_j$ with $nb(w_j) \geq 3$ to finally obtain an integer $N$ such that 
$$
\mathcal{R}_{N} = \{ \rho_{p,n} : p \in \N, n \geq N \}
$$
is included in a set of the form $z_1^* \ldots z_{q}^*$ where, for all $i$, $nb(z_i) \leq 2$.

% inversion des deux paragraphes ===
It remains to show that this observation leads to a contradiction. Let $\rho_{p,n}\in\mathcal{R}_{N}$. It can be factorized as $z_1^{n_1}\cdots z_{q}^{n_{q}}$. We have already observed that if $nb(z_i)=2$, then $n_i=1$. Otherwise, $nb(z_i)=1$ and thus $nb(z_{i}^{n_i})=1$. For this reason, we obtain that for all $n \geq N$, 
$$
nb(\rho_{p,n}) \leq 2q,
$$
which is a contradiction because $nb(\rho_{p,n})=n$ and this concludes the proof.
\end{proof}

\section{Conclusions}

As we have seen, there is a simple switch operation given by $\equiv_2$ that permits us to easily describe the $2$-binomial equivalence class of a word over a binary alphabet. 
One could try to generalize this operation over larger alphabets or for $k \geq 3$, but the question has no clear answer yet.

However, over a larger alphabet, we gave algorithmic and algebraic descriptions of the $2$-binomial classes. A natural question is to extend these results for $k \geq 3$.

We proved that $\LL(\sim_k,\alphabet)$ is not context-free if $k \geq 2$ and $\# \alphabet \geq 3$. We know that $\LL(\sim_2,\{1,2\})$ is context-free. However, the question is still open about $\LL(\sim_k,\{1,2\})$ with $k \geq 3$. It seems that a method similar to the one carried in Section~\ref{sec:5} could work, but it remains to find an unbounded set of singletons.

When $\LL(\sim_k,\alphabet)$ is not context-free, a measure of descriptional complexity is the so-called automaticity \cite{automaticity}. 
Let $L$ be a language and $C,t$ be integers. The idea is that we only know the words of $L$ of length at most $C$. Consider the following approximation of Nerode congruence: for any two words $u,v$ such that $|u|,|v| \leq t$,
$$
u \approx_{L,C,t} v \; \Leftrightarrow  \; \left( u^{-1} \left(L \cap \alphabet^{\leq C} \right)  \right) \cap \alphabet^{\leq C-t} = \left( v^{-1} \left(L \cap \alphabet^{\leq C} \right)  \right) \cap \alphabet^{\leq C-t}.
$$
The quantity $\# \left(\alphabet^{\leq  t}  /\!\approx_{L,C,t} \right)$ gives a lower approximation of the automaticity of $L$.
For $L = \LL(\sim_3, \{1,2\})$, $C=15$ and $t= 1, 2, \ldots, 9$, the first few values are
$$
1,3,5,9,16,27,49,88,154.
$$
For $L = \LL(\sim_2, \{1,2,3\})$, $C=9$ and $t= 1, 2, \ldots, 6$, they are
$$
1,4,8,19,42,62.
$$
Can the automaticity of such languages be characterized or estimated?


\begin{thebibliography}{99}

  \bibitem{CKPW} J. Cassaigne, J. Karhum\"aki, S. Puzynina, and M.A. Whiteland, $k$-abelian equivalence and rationality, {\em Fund. Infor.} {\bf 154} (2017), 65--94.

  \bibitem{FR} S. Foss\'e, G. Richomme, Some characterizations of Parikh matrix equivalent binary words, {\em Inform. Process. Lett.} {\bf 92} (2004), 77--82. 

  \bibitem{Manea} D. D. Freydenberger, P. Gawrychowski, J. Karhumäki, F. Manea, W. Rytter, Testing $k$-binomial equivalence, in {\em Multidisciplinary Creativity: homage to Gheorghe Paun on his 65th birthday}, 239--248, Ed. Spandugino, Bucharest, Romania (2015).

  \bibitem{Gaw} P. Gawrychowski, D. Krieger, N. Rampersad, J. Shallit, Finding the Growth Rate of a Regular of Context-Free Language in Polynomial Time, In: Ito M., Toyama M. (eds) Developments in Language Theory. DLT 2008, {\em Lect. Notes in Comp. Sci.} {\bf 5257} (2008). Springer, Berlin, Heidelberg

  \bibitem{GS}  S. Ginsburg, E. Spanier, Bounded ALGOL-like languages, {\em Trans. Amer. Math. Soc.} {\bf 113} (1964), 333--368.

  \bibitem{CF} R. Incitti, The growth function of context-free languages, {\it 
Theoret. Comput. Sci.} {\bf 255} (2001), 601--605.   

\bibitem{Kar} J. Karhum\"aki, Generalized Parikh mappings and homomorphisms, {\em Inform. and Control} {\bf 47} (1980), 155--165.

\bibitem{singleton} J. Karhum\"aki, S. Puzynina, M. Rao, M. Whiteland, On cardinalities of $k$-abelian equivalence classes, {\em Theoret. Comput. Sci.} {\bf 658} (2017), 190--204. 

\bibitem{KS1} J. Karhum\"aki, A. Saarela, L. Q. Zamboni, On a generalization of Abelian equivalence and complexity of infinite words, {\em J. Combin. Theory Ser. A} {\bf 120} (2013), 2189--2206.

  \bibitem{KS2} J. Karhum\"aki, A. Saarela, L. Q. Zamboni, Variations of the Morse-Hedlund theorem for $k$-abelian equivalence, {\em Lect. Notes in Comput. Sci.} {\bf 8633} (2014), 203--214.

\bibitem{Lejeune} M. Lejeune, {\em Au sujet de la complexité $k$-binomiale}, Master thesis,  University of Li\`ege (2018), \url{http://hdl.handle.net/2268.2/5007}.

\bibitem{RigoSalimov:2015} M. Rigo, P. Salimov, Another generalization of abelian equivalence: binomial complexity of infinite words, {\em Theoret. Comput. Sci.} {\bf 601} (2015), 47--57.

\bibitem{Salomaa}  A. Salomaa, Criteria for the matrix equivalence of words, {\em Theoret. Comput. Sci.} {\bf 411} (2010), 1818--1827.

\bibitem{automaticity} J. Shallit, Y. Breitbart, Automaticity. I. Properties of a measure of descriptional complexity, {\em J. Comput. System Sci.} {\bf 53} (1996), no. 1, 10--25.

\bibitem{Whit} M. A. Whiteland, {\em On the k-Abelian Equivalence
Relation of Finite Words}, Ph.D. Thesis, TUCS Dissertations {\bf 241}, Univ. of Turku (2019).
 

\end{thebibliography}
\end{document}